\documentclass{article}

\usepackage{amsmath,amsthm,amssymb,ytableau,graphicx,youngtab,paralist}
\usepackage{dsfont}
\usepackage{mathrsfs}

\usepackage[margin=1in]{geometry}
\usepackage{verbatim}
\usepackage{url}

\usepackage{tikz}
\usetikzlibrary{matrix,decorations.pathreplacing}

\swapnumbers
\numberwithin{equation}{section}
\newtheorem{theorem}[equation]{Theorem}
\newtheorem{lemma}[equation]{Lemma}
\newtheorem{proposition}[equation]{Proposition}
\newtheorem{conjecture}[equation]{Conjecture}
\newtheorem{corollary}[equation]{Corollary}

\newtheorem{claim}[equation]{Claim}

\theoremstyle{definition}

\newtheorem{remark}[equation]{Remark}
\newtheorem{definition}[equation]{Definition}
\newtheorem{example}[equation]{Example}
\newcommand{\rightbox}{\hfill\scalebox{0.8}{$\blacksquare$}}

\def\la{\lambda}

\def\<{\langle}
\def\>{\rangle}

\def\SL{ {\text {\rm SL} } }

\def\0{{\mathbf 0}}

\def\sgn{{\rm sgn}}

\newcommand{\dcmax}{{\textup{dc$_{\textup{max}}$}}}
\newcommand{\dcorig}{\textup{dc}}
\newcommand{\IC}{\mathbb C}
\newcommand{\IR}{\mathbb R}
\newcommand{\IN}{\mathbb N}

\DeclareMathOperator{\GL}{GL}
\DeclareMathOperator{\mult}{mult}
\DeclareMathOperator{\St}{St}
\DeclareMathOperator{\Sym}{Sym}
\renewcommand{\det}{\textup{det}}
\newcommand{\per}{\textup{per}}
\newcommand{\aS}{\ensuremath{\mathfrak{S}}}

\newcommand{\VPs}{\textup{\textsf{VP$_s$}}}
\newcommand{\VNP}{\textup{\textsf{VNP}}}
\newcommand{\sharpP}{\textup{\textsf{\#P}}}
\newcommand{\PP}{\textup{\textsf{P}}}
\newcommand{\NP}{\textup{\textsf{NP}}}
\newcommand{\exceptions}{{\mathscr{X}}}

\title{Rectangular Kronecker coefficients and plethysms in geometric complexity theory}
\author{Christian Ikenmeyer\footnote{Max Planck Institute for Informatics, Saarland Informatics Campus, cikenmey$@$mpi-inf.mpg.de. This work was done mainly while the first author was at Texas A\&M University.}, Greta Panova\footnote{University of Pennsylvania, panova$@$math.upenn.edu. Partially supported by NSF grant DMS-1500834.}}
\begin{document}
\sloppy
\maketitle

\begin{abstract}
We prove that in the geometric complexity theory program
the vanishing of rectangular Kronecker coefficients cannot be used to
prove superpolynomial determinantal complexity lower bounds for the permanent polynomial.

Moreover, we prove the positivity of rectangular Kronecker coefficients for a large class of partitions
where the side lengths of the rectangle are at least quadratic in the length of the partition.
We also compare rectangular Kronecker coefficients with their corresponding plethysm coefficients,
which leads to a new lower bound for rectangular Kronecker coefficients.
Moreover, we prove that the saturation of the rectangular Kronecker semigroup is trivial,
we show that the rectangular Kronecker positivity stretching factor is 2 for a long first row,
and we completely classify the positivity of
rectangular limit Kronecker coefficients that were introduced by Manivel in 2011.
\end{abstract}

{\footnotesize\noindent\textbf{MSC2010:} 20C30, 20G05, 68Q17}

{\footnotesize\noindent\textbf{Keywords:} Kronecker coefficients, plethysm coefficients, geometric complexity theory, positivity}

\setcounter{tocdepth}{1}
\tableofcontents

\section{Geometric complexity theory and Kronecker positivity}

The flagship problem in algebraic complexity theory is the determinant vs permanent problem,
as introduced by Valiant \cite{Val:79b}.
For $n \in \IN$ the polynomial
\[
\det_n := \sum_{\pi \in \aS_n} \sgn(\pi) X_{1,\pi(1)} \cdots X_{n,\pi(n)}
\]
is the well-known determinant polynomial, while for $m \in \IN$
\[
\per_m := \sum_{\pi \in \aS_m} X_{1,\pi(1)} \cdots X_{m,\pi(m)}
\]
is the permanent polynomial, a polynomial of interest in particular in graph theory and physics.
From a complexity theory standpoint the determinant is complete for the complexity class $\VPs$ \cite{Val:79b,toda:92,mapo:08},
while the permanent is complete for $\VNP$ \cite{Val:79b} and also for $\sharpP$ \cite{val:79}.
By definition we have $\VPs \subseteq \VNP$ and a major conjecture in algebraic complexity theory
related to the famous $\PP \neq \NP$ conjecture (see \cite{Coo:00}) is the following.
\begin{conjecture}\label{conj:flagshipconj}
$\VPs \neq \VNP$.
\end{conjecture}
This conjecture can be phrased independently of the definition of these complexity classes
as a question about expressing permanents as determinants of larger matrices as follows.
Valiant showed that for every polynomial $f$
there exists an integer $n \in \IN$
such that $f$ can be written as a determinant of an $n \times n$ matrix whose entries are
affine linear forms in the variables of $f$.
The smallest such number $n$ is called the \emph{determinantal complexity} of $f$, denoted by $\dcorig(f)$.
For example,
\[
\det\begin{pmatrix}
X_1 & 1 + X_2 \\
X_1 - X_2 & 1
\end{pmatrix} = X_1 +(1+X_2)(X_1-X_2) = X_1 X_2 - X_2^2 + 2 X_1 - X_2,
\]
so $\dcorig(X_1 X_2 - X_2^2 + 2 X_1 - X_2) \leq 2$.
Conjecture \ref{conj:flagshipconj} can be equivalently stated as follows:
\begin{conjecture}
The sequence $\dcorig(\per_m)$ grows superpolynomially in $m$.
\end{conjecture}
Finding lower bounds for $\dcorig(\per_m)$
is an important research area in algebraic complexity theory,
see for example the recent progress in \cite{MR:04, CCD:10, lamare:10, HI:14, abv:15, Yabe:15}.
Mulmuley and Sohoni \cite{muso:01,gct2} proposed an approach to this problem
using algebraic geometry and representation theory
and coined the term geometric complexity theory.

\subsection{Complexity lower bounds via representation theory}\label{subsec:compllowboundsrepth}
In the following we outline how one can prove lower bounds on $\dcorig(\per_m)$ using rectangular Kronecker coefficients,
see \eqref{eq:positivity} below.

A \emph{partition} $\la$ of $N$, written $\la \vdash N$, is defined to be a finite nonincreasing sequence of positive integers $\la=(\la_1,\la_2,\ldots,\la_{\ell(\la)})$,
where the \emph{length} $\ell(\la)$ denotes the number of entries in $\la$
and $\sum_{i=1}^{\ell(\la)}\la_i=N$. We write $|\la| := N$.
To a partition $\la$ we associate its \emph{Young diagram}, which is a top-aligned and left-aligned array of boxes
such that in row $i$ we have $\la_i$ boxes. Thus for $\la \vdash N$ the corresponding Young diagram has $N$ boxes.
For example, for $\la=(6,6,3,2,1,1)$ the associated Young diagram is
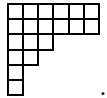
\[
\tiny\yng(6,6,3,2,1,1).
\]
If we transpose a Young diagram at the main diagonal we obtain another Young diagram, which we call $\la^t$.
The row lengths of $\la^t$ are the column lengths of $\la$.
In the example above we have $\la^t=(6,4,3,2,2,2)$.
For natural numbers $n$ and $d$ let $n \times d$ denote the partition $(d,d,\ldots,d)$,
i.e, the partition whose Young diagram is rectangular with $n$ rows and $d$ columns.
For two partitions $\la$ and $\mu$ let $\la+\mu$ denote the rowwise sum.
Moreover, for an integer $a$ let $a\la$ denote the rowwise scaling by $a$.
For $N \in \IN$ let $\aS_N$ denote the symmetric group on $N$ symbols.
For a partition $\la \vdash N$ let $[\la]$ denote the irreducible $\aS_{N}$-representation of type $\la$.
For partitions $\la,\mu,\nu$ of $nd$ let $g(\la,\mu,\nu) \in \IN$ denote the \emph{Kronecker coefficient}, i.e.,
the multiplicity of the irreducible $\aS_{nd}$-representation $[\la]$ in the tensor product $[\mu] \otimes [\nu]$,
where $[\mu] \otimes [\nu]$ is interpreted as an $\aS_{nd}$-representation via the diagonal embedding $\aS_{nd} \hookrightarrow \aS_{nd} \times \aS_{nd}$,
$\pi \mapsto (\pi,\pi)$.
A combinatorial interpretation of $g(\la,\mu,\nu)$ is known only in special cases,
see \cite{lasc:80, remm:89, remm:92, rewe:94, rosa:01, baor:06, bla:12, liu:14, IMW:15, Hay:15},
and finding a general combinatorial interpretation is problem~10 in Stanley's list of
positivity problems and conjectures in algebraic combinatorics \cite{sta:00}.
In geometric complexity theory
the main interest is focused on
\emph{rectangular} Kronecker coefficients,
i.e., the coefficients $g(\la,n \times d,n\times d)$.
These will be the main objects of study in this paper, as it was conjectured by Mulmuley and Sohoni \cite{gct2} that
their vanishing behaviour could be used to separate $\VPs$ from $\VNP$.
This conjecture spiked interest in these coefficients and has already led to several publications
inspired by geometric complexity theory.
Our main result says that it is impossible to separate $\VPs$ from $\VNP$ in this way.

For a partition $\la \vdash dn$ and a vector space $V$ of dimension at least $\ell(\la)$ let $\{\la\}$
denote the irreducible $\GL(V)$-representation of type $\la$.
The \emph{plethysm coefficient} $a_{\la}(d[n])$ is
the multiplicity of $\{\la\}$ in the $\GL(V)$-representation $\Sym^d(\Sym^n(V))$,
where $\Sym^\bullet$ denotes the symmetric power, i.e.,
$\Sym^n(V)$ can be identified with the vector space of homogeneous degree $n$ polynomials in $\dim(V)$ variables.
Analogous to the situation for the Kronecker coefficient,
finding a general combinatorial interpretation for $a_\la(d[n])$ is a fundamental open problem,
listed as problem~9 in \cite{sta:00}.

Not much is known about the third quantity we use, which is specialized to the permanent.
For fixed $n$ and $m$, $n > m$, let $\per_m^n := (X_{1,1})^{n-m}\per_m \in \Sym^n \IC^{n^2}$ denote the \emph{padded permanent polynomial}
(this is not the standard definition found in the literature, but it gives the same main result, see the appendix~\ref{subsec:additionalvar}).
Let $\GL_{n^2}\per_m^n$ denote its orbit and $\overline{\GL_{n^2}\per_m^n}$
its orbit closure (Zariski or Euclidean), which is an affine subvariety of the ambient space $\Sym^n\IC^{n^2}$.
For $\la \vdash dn$ let
$q^m_\la(d[n])$ denote the multiplicity of $\{\la\}$ in the homogeneous degree $d$
component of the coordinate ring $\IC[\overline{\GL_{n^2}\per_m^n}]$.
Since the orbit closure is a subvariety of the ambient space we have
\begin{equation}\label{eq:omegaa}
q^m_\la(d[n]) \leq a_\la(d[n]).
\end{equation}

\begin{theorem}[{see for example \cite[eq.~(5.2.7)]{BLMW:11}}]\label{thm:strategy}
If $q^m_\la(d[n]) > g(\la,n \times d, n\times d)$, then $\dcorig(\per_m)>n$.
\end{theorem}
Kronecker coefficients are $\sharpP$-hard to compute
as they are generalizations of the well-known
Littlewood-Richardson coefficients \cite{nara:06}.
But the positivity of Littlewood-Richardson coefficients can be decided in polynomial time \cite{deloera:06,MNS:12},
even by a combinatorial max-flow algorithm \cite{bi:12}.
Even though deciding positivity of Kronecker coefficients is $\NP$-hard in general, see \cite{IMW:15},
the same paper provides evidence that the rectangular Kronecker coefficient case is significantly simpler.
Thus to implement Thm.~\ref{thm:strategy} \cite{gct2} proposed to focus on the \emph{positivity} of representation theoretic multiplicities
in order to use the weaker statement
\begin{equation}\label{eq:positivity}
\text{If $q^m_\la(d[n]) > 0 = g(\la,n \times d, n\times d)$, then $\dcorig(\per_m)>n$}
\end{equation}
to prove lower bounds, see also \cite[Sec.~6.6]{Lan:13} and \cite[Problem~3.13]{Bur:15}, where this approach is explained.
The results in \cite{bci:09, BHI:15, Kum:15} already indicate that vanishing of $g(\la,n \times d, n\times d)$ might be rare.
In \cite{IMW:15} sequences of partition triples $(\la,\mu,\mu)$ are constructed that satisfy $g(\la,\mu,\mu)=0$.
Unfortunately $\mu$ has not the necessary rectangular shape.
Indeed, our main result Thm.~\ref{thm:together} completely rules out the possibility
that \eqref{eq:positivity} could be used to prove superpolynomial lower bounds on $\dcorig(\per_m)$.
\begin{theorem}[Main Theorem]\label{thm:together}
Let $n>3m^4$, $\la \vdash nd$. If $g(\la,n \times d, n \times d)=0$, then $q^m_\la(d[n])=0$.
\end{theorem}
Thm.~\ref{thm:together} holds in higher generality:
In the proof of Thm.~\ref{thm:together} we do not use any specific property of the permanent other than it is a
family of polynomials whose degree and number of variables is polynomially bounded in $m$.
For all these families of polynomials no superpolynomial lower bounds on the determinantal complexity can be shown
using the vanishing of $g(\la,n \times d, n \times d)$.

To use \eqref{eq:positivity} it is required by \eqref{eq:omegaa} that
\begin{align}
a_\la(d[n]) > 0 = g(\la,n \times d, n\times d).
\end{align}
An example is $\la=(13,13,2,2,2,2,2)$ and $d=12$, $n=3$, where we have $a_\la(d[n]) =1 > 0 = g(\la,n \times d, n\times d)$, see \cite[Appendix]{Ike:12b}.

We remark that---inspired by our paper---the very recent paper \cite{BIP:16} generalizes Thm.~\ref{thm:together} by
showing $q^m_\la(d[n])=0$ even if we only require that $\la$ does not occur in the
coordinate ring $\IC[\overline{\GL_{n^2}\det_n}]$ of orbit closure of the determinant (although they require a much larger $n$ in terms of $m$).
That disproves a stronger conjecture by Mulmuley and Sohoni.
While our paper uses representation theory to prove the main result, \cite{BIP:16} uses a more geometric approach.
We discuss more similarities and differences to the proof in \cite{BIP:16} at the end of this section, after
explaining the proof idea for Thm.~\ref{thm:together}.

A natural generalization of our main theorem would be to consider \emph{symmetric} rectangular Kronecker coefficients
instead of just $g(\la,n \times d, n\times d)$, because those are the multiplicities in the coordinate ring of the orbit $\GL_{n^2}\det_n$.
Our proof does not immediately work for those coefficients as they lack the important transposition symmetry,
but an equivalent no-go result for these coefficients follows directly from \cite{BIP:16}.

For a partition $\lambda$ we write $\bar\lambda$ to denote $\lambda$ with its first row removed, so $|\bar\lambda|+\la_1=|\la|$.
Vice versa, for a partition $\rho$ and $N \geq \rho_1$ we write $\rho(N) := (N - |\rho|,\rho)$ to denote the partition $\rho$ with an additional new first row
containing $N-|\rho|$ boxes.
The following theorem gives a very strong restriction on the shape of the partitions $\la$ that we have to consider.
\begin{theorem}[{\cite{LK:12}}]\label{thm:kadishlandsberg}
If $q^m_\la(d[n]) > 0$, then $|\bar\lambda|\leq md$ and $\ell(\la)\leq m^2$.
\end{theorem}
We sometimes write $\la_1 \geq d(n-m)$ instead of $|\bar\lambda|\leq md$.

The main ingredients for the proof of Thm.~\ref{thm:together} are the following Corollary~\ref{cor:degreelowerbound} and Thm.~\ref{thm:secondmain}.
\begin{corollary}[Degree lower bound]\label{cor:degreelowerbound}
If $|\bar\la|\leq md$ with $a_\la(d[n]) > g(\la,n \times d, n\times d)$, then $d>\frac n m$.
\end{corollary}
Corollary~\ref{cor:degreelowerbound} is proved in Section~\ref{sec:degreelowerbound}.

Note that if $n>3 m^4$ and $d > \frac n m$, then $d>3 m^3$.
\begin{theorem}[Kronecker positivity]\label{thm:secondmain}
Let $\exceptions$ be the set of the following 6 exceptional partitions.
\[
\exceptions := \{(1),(2\times1),(4\times1),(6\times1),(2,1),(3,1)\}.
\]
Let $d\in\IN$, $n \in \IN$, $\la \vdash dn$.\\
(a) If $\bar\lambda\in\exceptions$, then $a_\la(d[n])=0$.\\
(b) If $\bar\lambda\notin\exceptions$ and there exists $m \geq 3$ such that 
$\ell(\la)\leq m^2$,
$|\bar\lambda|\leq md$,
$d > 3 m^3$, and
$n > 3 m^4$, then $g(\la,n\times d, n \times d)>0$.
\end{theorem}
Thm.~\ref{thm:secondmain} is proved in Section~\ref{sec:secondthm}.
The proof cuts the partition $\la$ into smaller pieces in several significantly different ways
and makes heavy use of the following three properties:
\begin{itemize}
\item The semigroup property: If for 6 partitions $\la,\mu,\nu,\la',\mu',\nu'$ of $N$ we have that
$g(\la,\mu,\nu)>0$ and $g(\la',\mu',\nu')>0$, then also $g(\la+\la',\mu+\mu',\nu+\nu')>0$,
where we interpret partitions as integer vectors in order to define the sum of two partitions.
\item The transposition property: $g(\la,\mu,\nu)=g(\la,\mu^t,\nu^t)=g(\la^t,\mu^t,\nu)=g(\la^t,\mu,\nu^t)$.
\item The square positivity: For all positive $k$ we have $g(k \times k,k\times k,k\times k)>0$.
\end{itemize}
The first property is easy to see if we interpret $g(\la,\mu,\nu)$ as the dimension of the $(\la,\mu,\nu)$-highest weight vector space in
the coordinate ring of $V \otimes V \otimes V$. Here the acting group is $\GL(V)\times\GL(V)\times\GL(V)$ and the
semigroup property follows from multiplying highest weight vectors.
The second and third property follow from the character theory of the symmetric group.
While the second property is immediate, the third property (in \cite{bb}) requires reduction to the alternating group and specific properties of characters of symmetric partitions, later generalized in \cite{ppv:14} and further in \cite{pp:14b}. 

\begin{proof}[Proof of Thm.~\ref{thm:together}]
The proof is now straightforward.
Let $n > 3m^4$ and $\lambda \vdash nd$ such that $g(\lambda, n \times d, n \times d)=0$.
If $|\bar\lambda| > md$ or $\ell(\lambda) > m^2$, then $q_\lambda^m(d[n])=0$ by Thm.~\ref{thm:kadishlandsberg} and we are done.
Thus we assume from now on that $|\bar\lambda| \leq md$ and $\ell(\lambda) \leq m^2$.
If $d \leq \frac n m$, then by Cor. 1.6 we have $a_\la(d[n])\leq g(\la,n\times d,n\times d)=0$.
By \eqref{eq:omegaa} it follows $q_\lambda^m(d[n])=0$ and we are done.
Thus we also assume from now that $d > \frac n m > 3m^3$.
From $g(\la,n\times d,n\times d)=0$ we conclude with Thm.~\ref{thm:secondmain}(b) that $\bar\la \in \exceptions$.
But Thm.~\ref{thm:secondmain}(a) implies that $a_\la(d[n])=0$, which by \eqref{eq:omegaa} implies $q_\lambda^m(d[n])=0$ and we are done.
\end{proof}
\subsection{Consequences for geometric complexity theory}\label{sec:consequences}
Algebraic geometry and representation theory guarantee that if $(X_{1,1})^{n-m}\per_m \notin \overline{\GL_{n^2}\det_n}$,
then there exists a homogeneous polynomial $P$ in an irreducible representation $\{\la\}\subseteq \Sym^d(\Sym^n\IC^{n^2})$ such that
$P$ vanishes on $\overline{\GL_{n^2}\det_n}$ and $P((X_{1,1})^{n-m}\per_m)\neq 0$.
One suggestive way of finding these $P$ was \eqref{eq:positivity}.
Although Thm.~\ref{thm:together} rules out this possibility,
there are several ways in which these $P$ could still be found.
One such way, namely finding irreducible $\GL_{n^2}$-representations in
the coordinate ring $\IC[\overline{\GL_{n^2}(X_{1,1})^{n-m}\per_m}]$ that do not occur in
the coordinate ring $\IC[\overline{\GL_{n^2}\det_n}]$ has been proved impossible \cite{BIP:16}.
Not much is known if we study the actual multiplicities of irreducible $\GL_{n^2}$-representations in the coordinate rings.

We remark that the overall proof structure in \cite{BIP:16} closely mimics the proof structure of our Thm.~\ref{thm:together}:
\cite{LK:12} is used to prove a degree lower bound and then an analog to Thm.~\ref{thm:secondmain} is proved,
also by cutting the partition $\la$ into smaller pieces.
The details and the methods used in \cite{BIP:16} are very different from our paper though.

\subsection{Positivity results for Kronecker coefficients}

In Section~\ref{sec:secondthm} we prove the rectangular Kronecker positivity for a large class of partitions:
If the side lengths of the rectangles are at least quadratic in the partition length we have positivity,
see Theorem~\ref{thm:mainpositivityresult} for the precise statement.
More exact Kronecker positivity results are given in Section~\ref{sec:morepositivity}.

Combinatorial conjectures like the Saxl conjecture \cite{ppv:14,Ike:15,LS:15} are also concerned with the positivity of Kronecker coefficients.

In Section~\ref{sec:corollaries} we we prove that the saturation of the rectangular Kronecker semigroup is trivial,
we show that the rectangular Kronecker positivity stretching factor is 2 for a long first row,
and we completely classify the positivity of
rectangular limit Kronecker coefficients.

\medskip {\bf Acknowledgments.} 

We gratefully acknowledge the helpful comments of the anonymous reviewers. GP was partially supported by NSF grant DMS-1500834.

\section{Proof of the degree lower bound}\label{sec:degreelowerbound}
In this section we prove the degree lower bound Cor.~\ref{cor:degreelowerbound}.
The main ingredients are Manivel's result about limit rectangular Kronecker coefficients (Section~\ref{subsec:manivel})
and Valiant's insights about finite determinantal complexity (Section~\ref{subsec:infgoodlow}).
\subsection{Stable rectangular Kronecker coefficients}\label{subsec:manivel}
The main contribution to the specific limits of Kronecker coefficients that we are interested in in this paper comes from Manivel.
\begin{theorem}[{\cite[Thm.~1]{man:11}}]\label{thm:manivel_rect}
Fix a partition $\rho$.
The function $g(\rho(nd), n \times d, n \times d)$ is symmetric and nondecreasing in $n$ and $d$.
If $n \geq |\rho|$ we have that 
$g(\rho(nd), n \times d, n \times d) = a_{\rho}(d),$
where $a_\rho(d)$ denotes the dimension of the $\SL_d$ invariant space $\big(S_\rho(S_{(2,1^{d-2})} \IC^d )\big)^{\SL_d}$,
where $S_\rho$ and $S_{(2,1^{d-2})}$ denote Schur functors.
\end{theorem}

\begin{remark}
Since $g(\rho(nd), n \times d, n \times d)$ is symmetric in $n$ and $d$,
if both $d \geq |\rho|$ and $n \geq |\rho|$, then
$g(\rho(nd), n \times d, n \times d) = a_\rho$ and $a_{\rho}$ depends only on $\rho$. \rightbox
\end{remark}

The pairs $(n,d)$ for which the Kronecker coefficient $g(\rho(nd), n \times d, n \times d)$ reaches its maximum form the \emph{stable range} $\St(\rho)$,
a monotone subset of $\mathbb{Z}^2$:
\[
\St(\rho) := \{(n,d)  \mid  g(\rho(nd) , n \times d, n\times d) = a_\rho \}.
\]

Analogously, the 1-stable range is defined as
\[
\St^1(\rho) := \{(n,d)  \mid  g(\rho(nd) , n \times d, n\times d) = a_\rho(d) \}.
\]

\begin{example}
The following tables show the Kronecker coefficients $g(\rho(nd),n \times d, n\times d)$ for $\rho=(6)$, $\rho=(2,1,1,1,1,1)$, and $\rho=(3,1,1,1,1)$,
from left to right.
\begin{center}
\begin{minipage}{11cm}
\begin{minipage}{3cm}
\begin{verbatim}
0 0 0 0 0 0 ...
0 0 0 0 0 1 ...
0 0 0 1 1 2 ...
0 0 1 2 2 3 ...
0 0 1 2 2 3 ...
0 1 2 3 3 4 ...
. . . . . . .
. . . . . .  .
. . . . . .   .
\end{verbatim}
\end{minipage}
\mbox{~~~~~}
\begin{minipage}{3cm}
\begin{verbatim}
0 0 0 0 0 0 0 ...
0 0 0 0 0 0 0 ...
0 0 1 1 1 1 1 ...
0 0 1 2 2 2 2 ...
0 0 1 2 2 2 2 ...
0 0 1 2 2 2 2 ...
0 0 1 2 2 2 2 ...
. . . . . . . .
. . . . . . .  .
. . . . . . .   .
\end{verbatim}
\end{minipage}
\mbox{~~~~~}
\begin{minipage}{3cm}
\begin{verbatim}
0 0 0 0 0 0 0 ...
0 0 0 0 0 0 0 ...
0 0 0 1 2 2 2 ...
0 0 1 3 4 4 4 ...
0 0 2 4 5 5 5 ...
0 0 2 4 5 5 5 ...
0 0 2 4 5 5 5 ...
. . . . . . . .
. . . . . . .  .
. . . . . . .   .
\end{verbatim}
\end{minipage}
\end{minipage}
\end{center}
\rightbox
\end{example}

\subsection{Complexity lower bounds literally too good to be true: inequality of multiplicities and degree lower bound}\label{subsec:infgoodlow}
In this section we prove Corollary~\ref{cor:degreelowerbound} by using the finiteness of the determinantal complexity.

Define the finite dimensional vector space $V_m := \Sym^m \IC^{m^2}$ of homogeneous degree $m$ polynomials in $m^2$ variables $X_{1,1}, X_{1,2}, \ldots, X_{m,m}$.
For $n\geq m$ and $f \in V_m$
let $f^{\sharp n} \in V_n$ denote the product $(X_{1,1})^{n-m} f$.
Let \[\Gamma_m^n := \overline{ \{ g f^{\sharp n} \mid g \in \GL_{n^2}, f \in V_m\}} \subseteq V_n\] denote the variety of padded polynomials.
We define $o^{m}_\la(d[n])$
to be the multiplicity of the irreducible $\GL_{n^2}$-representation $\{\la\}$ in the coordinate ring $\IC [ \Gamma_m^n ]$.
In the notation from Section~\ref{subsec:compllowboundsrepth} we have
$q^m_\la(d[n])\leq o^{m}_\la(d[n])$.
Since $\Gamma_m^n \subseteq V_n$ is an affine subvariety, we have
$
0 \leq o_\la^m(d[n]) \leq a_\la(d[n]),
$
where $d := |\la|/n$. If $|\la|$ is not divisible by $n$, then $0 = o_\la^m(d[n]) = a_\la(d[n])$.

\begin{definition}[$m$-obstruction of quality $n$]\label{def:f_obstr}
An $m$-obstruction of quality $n$ is defined to be a partition $\la$ such that $|\lambda|=nd$ for some $d \in \mathbb N$
and
$
g(\la,n\times d, n \times d) < o_{\la}^m(d[n]).
$
\end{definition}

As the following proposition shows, the existence of an $m$-obstruction of quality $n$
proves the existence of an $f\in V_m$ that cannot be written as an $n \times n$ determinant.
\begin{proposition}\label{prop:peter-weyl}
If there exists an $m$-obstruction of quality $n$, then $\dcorig(f) > n$ for some $f \in V_m$.
\end{proposition}
\begin{proof}
By the algebraic Peter-Weyl theorem \cite[eq.~(5.2.7)]{BLMW:11} it follows:
$
\mult_\la (\IC[\overline{\GL_{n^2}\det_n}]) \leq g(\la,n\times d, n \times d).
$
Since an $m$-obstruction of quality $n$ exists, it follows that
\begin{equation*}\tag{$\ast$}
\mult_\la (\IC[\overline{\GL_{n^2}\det_n}]) < o_{\la}^m(d[n]).
\end{equation*}
Assume for the sake of contradiction that $\dcorig(f)\leq n$ for all $f \in V_m$.
This implies that $\Gamma_m^n \subseteq \overline{\GL_{n^2}\det_n}$ as an affine subvariety.
The restriction of functions is a $\GL_{n^2}$ equivariant surjection between the homogeneous degree $d$ parts of the coordinate rings:
$
\IC[\overline{\GL_{n^2}\det_n}]_d \twoheadrightarrow \IC[\Gamma_m^n]_d.
$
By Schur's lemma it follows
$
\mult_\la (\IC[\overline{\GL_{n^2}\det_n}]) \geq \mult_\la(\IC[\Gamma_m^n]) = o_\la^m(d[n]),
$
which is a contradiction to $(\ast)$.
\end{proof}

\begin{proposition}[{\cite{LK:12}}]\label{prop:kadish-landsberg}
\begin{enumerate}
\item[(a)] Given $\la \vdash nd$, if $|\bar\la|>dm$, then $o_\la^m(d[n]) = 0$.
\item[(b)] Given $\la \vdash md$, we have $o_{\la+(dn-dm)}^m(d[n]) \geq a_\la(d[m])$.
\end{enumerate}
\end{proposition}
\begin{proof}
Part (a) is the first part of \cite[Thm 1.3]{LK:12}.
Part (b) follows from the proof of the second part of \cite[Thm 1.3]{LK:12}:
From a highest weight vector $P$ of weight $\la\vdash md$ in $\Sym^d(V_m)$ they construct
a highest weight vector $P^{\sharp n}$ of weight $\la+(d(n-m))\vdash nd$ in $\Sym^d(V_n)$ such that
the evaluations $P(f) = P^{\sharp n}(f^{\sharp n})$ coincide.
Take a basis $P_1,\ldots,P_{a_\la(d[m])}$ of the highest weight vector space of weight $\la$ in $\Sym^d(V_m)$
and take general points $f_1,\ldots,f_{a_\la(d[m])}$ from $V_m$.
Then the evaluation matrix $\big( P_i(f_j) \big)_{1\leq i,j\leq a_\la(d[m])}$ has full rank.
Thus $\big( P^{\sharp n}_i(f^{\sharp n}_j) \big)_{1\leq i,j\leq a_\la(d[m])}$ has full rank,
which implies the statement.
\end{proof}

Let $\dcmax(m)$ denote the maximum $\max\{\dcorig(f) \mid f \in V_m\}$.
The following lemma shows that $\dcmax(m)$ is a well-defined finite number.

\begin{lemma}\label{lem:finite_dc}
Fix $m \in \mathbb N$.
There is a number $\dcmax(m) \in \mathbb N$ such that
for all $f \in V_m$
we have $\dcorig(f)\leq \dcmax(m)$.
\end{lemma}
\begin{proof}
From \cite{Val:79b} it follows that if $f$ has $r$ monomials, then $\dcorig(f)\leq rm$.
In particular, since every $f \in V_m$ has at most $\binom{m+m^2-1}{m}$ monomials, it follows that
$\forall f \in V_m : \ \dcorig(f) \leq \binom{m+m^2-1}{m}m$.
\end{proof}

\begin{proposition}[Inequality of Multiplicities]\label{pro:main}
Fix $\rho$, and let $(n,d) \in \St^1(\rho)$,
which is true in particular if $n \geq |\rho|$.
Let $\la = \rho(nd)$. Then 
$
g( \la , n \times d, n \times d) \geq a_\la(d[n]).
$
\end{proposition}
Interestingly the proof of this purely representation theoretic statement uses the finiteness of the determinantal complexity.
\begin{proof}[Proof of Prop.~\ref{pro:main}]
Assume the contrary, i.e., there exists a $\rho$ and $(m,d) \in \St^1(\rho)$ with
$
g(\rho(md), m \times d, m \times d) < a_{ \rho(md)}(d[m]).
$
Since $(m,d) \in \St^1(\rho)$, the Kronecker coefficient $g(\rho(nd),n\times d,n \times d)$ is the same for all $(n,d)$ with $n\geq m$.
Thus we have 
$
g(\rho(n d), n \times d, n \times d) <  a_{\rho( md )}(d[m])
$
for all $n$ with $n\geq m$.
By Prop.~\ref{prop:kadish-landsberg}(b) we have $o_{\rho(n d)}^{m}(d[n])\geq a_{\rho( md )}(d[m])$
for all $n \geq m$.
This implies $\rho(nd)$ is an $m$-obstruction of quality $n$ for all $n \geq m$.
By Prop.~\ref{prop:peter-weyl} there exist functions $f_n \in V_m$
such that $\dcorig(f_n) > n$ for every $n > m$.
In particular we have
$
\dcorig(f_{\dcmax(m)}) > \dcmax(m),
$
in contradiction to Lem.~\ref{lem:finite_dc}.
\end{proof}

\begin{proof}[Proof of Corollary~\ref{cor:degreelowerbound}]
The proof is now immediate.
By assumption we have $|\bar\la|\leq md$. By Prop.~\ref{pro:main} we have
$(n,d)\notin\St^1(\bar\la)$, so $|\bar\la|>n$.
Thus $dm \geq |\bar\la| > n$ and therefore $d > \frac n m$.
\end{proof}

\section{Kronecker positivity: A simplified version}\label{sec:didact}
In this section we prepare the reader for the combinatorially intricate arguments to come in the proof of Thm.~\ref{thm:secondmain}.
This section is a purely didactical one. We prove a weaker statement (Thm.~\ref{thm:didact}) than Thm.~\ref{thm:secondmain} in the following sense.
The bound we obtain is weaker and we exclude column lengths 2, 3, 5, and 7 from $\la$, which are the column lengths that cause most of the technical issues.
The paper is self-contained without this section \ref{sec:didact}. Neither Thm.~\ref{thm:didact} nor its proof are referenced later.
On the contrary, to prove Thm.~\ref{thm:didact} we use two positivity results (Lem.~\ref{lem:rect_pos} and Cor.~\ref{cor:hooks1}) that will be proved in section~\ref{sec:secondthm}.
\begin{theorem}\label{thm:didact}
Let $m\geq 3$ and let $n>m^9$ and $d > m^8$.
Let $\la \vdash nd$ with $|\bar\la|\leq md$ and $\ell(\la)\leq m^2$.
If $\la$ has no column of length 2, 3, 5, or 7, then $g(\la,n \times d, n \times d)>0$.
\end{theorem}
Let $(i,1^j)$ denote the hook partition of $i+j$ with $j+1$ boxes in the first column.
\begin{lemma}\label{lem:didacthooks}
Let $m\geq 3$ and $d \geq m^2$.
Let $1 \leq j < m^2$, $j \notin \{1,2,4,6\}$.
Then $g((d m^2-j,1^j),d \times m^2, d \times m^2) > 0$.
\end{lemma}
\begin{proof}
This follows from Cor.~\ref{cor:hooks1}, because $m^4-6 > m^2$.
\end{proof}
Together with Lem.~\ref{lem:rect_pos} we now have all the building blocks to prove Thm.~\ref{thm:didact}.
\begin{proof}[Proof of Thm.~\ref{thm:didact}]
We forget about the first row of $\la$ and treat each column length $k \in \{2,\ldots,m^2\}$ in $\la$ separately.
For the ease of notation let $\bar k := k-1$.
For each $k$ let $c_k$ be the number of columns of length $k$ in $\la$.
We divide $c_k$ by $k$, formally $c_k = x_k k + r_k$ with $r_k < k$.
The partition $\bar\la$ decomposes as follows:
\[
\bar\la = \sum_{k=2}^{m^2} x_k (\bar k\times k) + \sum_{k=2}^{m^2} r_k (\bar k \times 1).
\]
Note that by assumption on $\la$ both sums exclude the four values $k \in \{2,3,5,7\}$.
We now group together the $(\bar k\times k)$ rectangles into groups of roughly $\frac d k$,
so that the result is roughly of size $\bar k\times d$.
Formally we define $s_k := \lfloor d/k\rfloor$ and divide $x_k = h_k s_k + t_k$ with $t_k < s_k$.
Now $\bar\la$ decomposes as follows:
\[
\bar\la = \sum_{k=2}^{m^2} h_k (\bar k\times k s_k)
+ \sum_{k=2}^{m^2} (\bar k\times k t_k)
+ \sum_{k=2}^{m^2} r_k (\bar k\times 1).
\]
We prove  Kronecker positivity for the summands separately.
For the leftmost sum, $g((\bar k\times k s_k)(dk),d \times k, d \times k)>0$ by Lem.~\ref{lem:rect_pos} with $a=d$.
For the middle sum, $g((\bar k\times k t_k)(dk),d \times k, d \times k)>0$ with the same argument.
For the rightmost sum, $g( (dm^2-\bar k,1^{\bar k}),d \times m^2, d \times m^2)>0$ by Lem.~\ref{lem:didacthooks}.
According to these observations we add a new first row to $\bar\la$. Formally $\mu:=$
\[
\sum_{k=2}^{m^2} h_k (\bar k\times k s_k)(dk)
+ \sum_{k=2}^{m^2} (\bar k\times k t_k)(dk)
+ \sum_{k=2}^{m^2} r_k (\bar k\times 1)(m^2 d).
\]
Using the semigroup property the above three positivity considerations we show that
\begin{equation}\label{eq:didact:almost}
g(\mu, d \times \tilde n, d \times \tilde n) > 0,
\end{equation}
where $\tilde n = \sum_{k=2}^{m^2}h_k k + \sum_{k=2}^{m^2}k + r_k m^2$.

We now show $\tilde n \leq n$ using very coarse estimates for the sake of simplicity.
Since $|\bar\la| \leq md$ we have $x_k \leq md$ for $2 \leq k \leq m^2$.
By definition,  $s_k \geq \lfloor \frac d {m^2} \rfloor \geq \frac d {m^3}$.
Thus $h_k \leq x_k/s_k \leq md/(d/m^3) \leq m^4$.
Since $r_k < k \leq m^2$ we have that
$\tilde n \leq m^4\cdot m^2\cdot m^2 + m^2 \cdot m^2 + m^3 \leq m^9 \leq n$.

Set $n' := n-\tilde n \geq 0$.
Since both $|\la|$ and $|\mu|$ are divisible by $d$
we can use the semigroup property and add the positive Kronecker triple
$((n' d),d \times n', d \times n')$ to \eqref{eq:didact:almost}.
Since $\bar\la=\bar\mu$ we conclude
$g(\la,n\times d, n\times d)>0$.
\end{proof}

\section{Proof of Kronecker positivity}\label{sec:secondthm}
In this section we prove Thm.~\ref{thm:secondmain}.
If $\bar\la \in \exceptions$, a finite calculation reveals $a_{\bar\la} = 0$.
Combining this with Prop.~\ref{pro:main} gives $a_\la(d[n])=0$.
Thus we are left to analyze the case $\bar\la \notin \exceptions$.

Given a partition $\nu \notin \exceptions$ we want to decompose it into smaller partitions and use the semigroup property to show
the positivity of $g(\nu(ab),a\times b,a\times b)$.
We use the following decomposition theorem to write $\nu$ as a sum of partitions that are not in $\exceptions$.
\begin{lemma}[Partition decomposition]\label{lem:partitiondecomposition}
Given $\nu \notin \exceptions$ let $\ell:=\ell(\nu)+1$.
We can find $x_k \in \IN$, $y_k \in \IN$, $2 \leq k \leq \ell$,
and a partition $\xi$ such that
\[\textstyle
\nu = \rho + \xi + \sum_{k=2}^\ell x_k ((k-1) \times k) + \sum_{k=2}^\ell y_k ((k-1) \times 2),
\]
with $y_k < k$ and
where all columns in $\xi$ have distinct lengths and no column in $\xi$ has length 1, 2, 4, or 6,
and where $\rho$ is of one of the following shapes:
\begin{enumerate}[(1)]
\item $\rho$ has only columns of length 1, 2, 4, 6, all column lengths are distinct, and $\rho \notin \exceptions$.
\item $\rho = (i \times 1)+\eta$, where $i \notin \{1,2,4,6\}$, $i \leq \ell-1$, and $\eta \in \exceptions \setminus \{(3,1)\}$.
\item $\rho = (i \times 2)+\eta$, where $i \in \{2,4,6\}$ and $\eta \in \exceptions \setminus \{(3,1)\}$.
\item $\rho = (4)+\eta$, where $\eta \in \exceptions \setminus \{(3,1)\}$.
\item $\rho \in \{(3,1,1,1,1,1),(3,1,1,1),(3),(4,1)\}$.
\end{enumerate}
\end{lemma}
\begin{proof}
We start by treating each $k$ independently.
Let $c_k$ denote the number of columns of length $k-1$ in $\nu$.
In a greedy manner cut off from $\nu$ as many rectangles of size $(k-1)\times k$ as possible.
Formally, we divide
$c_k = x'_k k + r'_k$ with $r'_k<k$.
We are left with a $(k-1) \times r'_k$ rectangle.
We now join (if possible) one of the $(k-1)\times k$ rectangles with the $(k-1) \times r_k$ rectangle:
If $x'_k\geq 1$, define $r_k := r'_k+k$ and $x_k:=x'_k-1$.
If $x'_k=0$, define $r_k := r'_k$ and $x_k:=x'_k$.
We obtain a $(k-1) \times r_k$ rectangle and call it $R_k$.
Note that $r_k < 2k$.

Cut off from $R_k$ rectangle as many rectangles of size $(k-1)\times 2$ as possible.
Formally, define $y_k := \lfloor r_k / 2 \rfloor$, so $y_k < k$ as required in the claim.
After cutting we are left with either the empty partition or a column $(k-1)\times 1$.
In the former case (i.e., $r_k$ is even) define $b_k:=0$, otherwise $b_k:=1$.

Looking at all $k$ together we define two remainder partitions
$
\rho := \sum_{k-1=1,2,4,6} b_k ((k-1)\times 1)
$
and
$
\xi := \sum_{k-1 \in [1,\ell-1]\setminus\{1,2,4,6\}} b_k ((k-1)\times 1).
$
Clearly
\[\textstyle
\nu = \rho + \xi + \sum_{k=2}^\ell x_k ((k-1) \times k) + \sum_{k=2}^\ell y_k ((k-1) \times 2),
\]
as required in the statement of the claim.
Moreover, $\xi$ has the correct shape.
Clearly $\rho$ has only columns of length 1, 2, 4, 6 and all column lengths are distinct.
If $\rho \notin \exceptions$, then we are done by property (1).

The rest of the proof is devoted to the case where $\rho \in \exceptions$.
We will use that if $y_k=0$, then $\nu$ has at most 1 column of length $k-1$, by definition of $R_k$.
Note that $\rho \neq (3,1)$, because no two columns in $\rho$ have the same length.
If $\nu$ has a column of length $i$ different from 1, 2, 4, 6, then
such a column appears in $\xi$ or we have that $y_{i+1} > 0$ and $\xi$ has no column of length $i$.
If it appears in $\xi$, then we remove it from $\xi$ and add it to $\rho$ and we are done by property (2).
If it does not appear in $\xi$ but $y_{i+1}>0$, then we decrease $y_{i+1}$ by 1 and add a column of length $i$ to both $\xi$ and $\rho$,
so that we are done by property (2).

So from now on we assume that $\nu$ only has columns of length 1, 2, 4, or 6 (and thus $x_k=0$ for all $k \notin \{2,3,5,7\}$).

If $y_{i+1}>0$ for some $i \in \{2,4,6\}$, then we can decrease $y_{i+1}$ by 1 and set $\rho \leftarrow \rho + (i \times 2)$,
so we are done by property~(3).

Thus from now on we assume $y_{i+1}=0$ for all $i \in \{2,4,6\}$.
Note that $y_2$ corresponds to the columns of length 1 in $\nu$.
If $y_2\geq 2$, then we can decrease $y_2$ by 2 and set $\rho \leftarrow \rho + (4)$,
so we are done by property~(4).

At this point the possible shapes of $\nu$ are quite limited:
$c_{i+1}=0$ for all $i \notin\{1,2,4,6\}$, $c_{i+1} \leq 1$ for $i \in \{2,4,6\}$, $c_{2}\leq 3$.

If $y_2 = 0$, then $\nu=\rho$ by construction, which is impossible because $\nu \notin \exceptions$ and $\rho \in \exceptions$.

Recall $\rho \in \exceptions \setminus\{(3,1)\}$.
If $y_2 = 1$, then $\nu = \rho+(2)$, so
$
\nu \in \{(3,1,1,1,1,1),(3,1,1,1),(3,1),(3),(4,1)\}.
$
Since $\nu \notin \exceptions$ it follows
$
\nu \in \{(3,1,1,1,1,1),(3,1,1,1),(3),(4,1)\}.
$
Now we set $y_2 \leftarrow 0$ and $\rho \leftarrow \nu$ and we are done by property~(5).
\end{proof}

All summands in the partition decomposition will yield a positive rectangular Kronecker coefficient,
but we will need to group large blocks of $(k-1)\times k$ rectangles.
This is done with the following lemma.
\begin{lemma}\label{lem:rect_pos}
Let   $\mu = (k \times (ks) )$ and $a \geq ks$, then $g( k \times (ks) + (k(a-ks)), a \times k, a \times k)>0$. 
\end{lemma}
\begin{proof}
For all $k$ we have $g( k \times k, k \times k, k \times k)>0$ as shown in \cite{bb}.
By the semigroup property we can add these square triples $s$ times and obtain
$g( k \times (ks), k \times (ks), k \times (ks) )>0.$
Let $(k(a-ks))$ denote the single row partition of $k(a-ks)$.
Clearly $g( (k(a-ks)), k \times (a-ks), k \times (a-ks) )>0$.
Adding these two triples with the semigroup property we get
$g( k \times (ks) + (k(a-ks)), k \times a, k \times a)>0.$
Finally, transposing the last two partitions
we obtain the statement.
\end{proof}

The following proposition will be used to prove positivity for building blocks of partitions like hooks and fat hooks. Here, for a set $S$ and a number $x$ we denote by $x - S$ the set $\{ x-y \mid y \in S\}$.

\begin{proposition}\label{prop:extensions}
Let $\rho$ be a partition of length $\ell$ or $\rho=\emptyset$ with $\ell=0$, and denote by $\nu^k := k \times 1 + \rho$ for $k\geq \ell$ and let $R_\rho := |\rho|+\rho_1+1$.
Suppose that there exists an integer $a  >  \max( \sqrt{R_\rho+\ell}+3,\frac{\ell}{2}-1,6)$
and subsets $H^1_\rho, H^2_\rho \subseteq [\max(\ell,1), 2a+1]$
such that
$g( \nu^k(a^2), a \times a, a \times a)>0$ for all $k \in [\ell , a^2 - R_\rho ] \setminus  \left( H^1_{\rho} \cup (a^2 - H^2_\rho) \right)$.
Then for every $b \geq a$ we have that
$g(\nu^k(b^2), b \times b, b \times b)>0$ for all  $k \in [\ell , b^2 - R_\rho ] \setminus \left( H^1_{\rho} \cup (b^2 - H^2_\rho) \right) $.
\end{proposition}
Note that $R_\rho$ is the size of the partition obtained from $\rho$ by adding a shortest possible top row and one extra box at the top left corner, so that the largest possible column we can add is $h^2 - R_\rho$ to still have a valid partition. 
\begin{proof}
The proof is by induction on $b$ and repeated application of the semigroup property.  We have that the statement is true for $b=a$ by the given condition. Assume that the statement holds for $b=c$ for some $c$, we show that it holds for $b=c+1$. We do this by showing the following claim which helps us extend from positivity for $b=c$ to $b=c+1$. Let $P_c := \{k: g( \nu^k(c^2), c \times c, c \times c) >0 \}$ be the set of values of $k$ for which the Kronecker coefficient is positive. The inductive step is equivalent to the statement: if $[ \ell , c^2 - R_\rho] \setminus \left( H^1_{\rho} \cup (c^2 - H^2_\rho) \right) \subset P_c$, then $[ \ell , (c+1)^2 - R_\rho] \setminus \left( H^1_{\rho} \cup ((c+1)^2 - H^2_\rho) \right) \subset P_{c+1}$.

\textbf{Claim:} Suppose that $k \in P_c$, then $k, k+2c+1 \in P_{c+1}$.

\textit{Proof of claim:} 
To show that $k \in P_{c+1}$
 we apply the semigroup property by successively adding triples $((x), x \times 1, x \times 1)$ for $x=c$ and then $x=c+1$ and transposing the rectangles.
 Let $\tilde c := c+1$.
\begin{eqnarray*}
0 &<&
g( \nu^k(c^2), c\times c, c\times c) \leq  g( \nu^k(c^2) + (c ), c\times \tilde c, c \times \tilde c) \\
&=& g( \nu^k(c^2) + (c) , \tilde c \times c, \tilde c \times c)\\
&\leq& g( \nu^k(c^2) + (c)+\tilde c , \tilde c \times {\tilde c}, {\tilde c} \times {\tilde c} ) \\
 &=& g(\nu^k( {\tilde c}^2), {\tilde c}\times {\tilde c}, {\tilde c}\times {\tilde c} ).
\end{eqnarray*}
Next, to show that $k +2c+1 \in P_{c+1}$ we first transpose $\nu^k(c^2)$ and one of the squares, apply the argument from above to it, and then transpose again:
\begin{eqnarray*} 
0 &<& g( \nu^{k} (c^2), c\times c, c\times c) = g( (\nu^{k} (c^2))^t  , c\times c, c\times c)\\
 &\leq& g( (\nu^k(c^2))^t + (c ), c\times {\tilde c}, c \times {\tilde c}) \\
 &=& g( (\nu^k(c^2))^t + (c) , {\tilde c} \times c, {\tilde c} \times c) \\
 &\leq& g( (\nu^k(c^2))^t + (2c+1)  , {\tilde c} \times {\tilde c}, {\tilde c} \times {\tilde c} ) \\
 &=& g(\nu^{k+2c+1}( {\tilde c}^2), {\tilde c}\times {\tilde c}, {\tilde c}\times {\tilde c} ).\quad\quad\quad\quad\quad\quad\quad\quad\square
\end{eqnarray*}

The claim implies that $P_c \cup (2c+1 +P_c) \subset P_{c+1}$.

Now we apply the claim to our sets.  By hypothesis, we have that $S_c:= [ \ell , c^2 - R_\rho ] \setminus \left( H^1_{\rho} \cup (c^2 - H^2_\rho) \right) \subset P_c$. 

Suppose that $S_{c+1} \not \subset P_{c+1}$, so there exists a $k \in S_{c+1}$, such that $k \not \in P_{c+1}$. By the claim, we must then have that both $ k \not \in P_c$ and $k-(2c+1) \not \in P_c$. Since $S_c \subset P_c$, this means that $k , k-(2c+1) \in (-\infty,\ell) \cup H_\rho^1 \cup (c^2 - H^2_\rho) \cup [c^2 -R_\rho+1, +\infty)$. Translating the sets by $2c+1$ and observing that $c^2+(2c+1)=(c+1)^2$, and using that $ k \in [\ell, (c+1)^2-R_\rho]$ already, we must have that 
\begin{eqnarray}\label{eq:sets}
 k \in \left(  H^1_\rho \cup (c^2 \! - \! H^2_\rho) \cup [c^2\!-\! R_\rho\!+\! 1,(c\!+\!1)^2\!-\!R_\rho] \right)  
 \\ \cap \underbrace{ \left( [\ell,2c\!+\!1\!+\!\ell]\cup  (2c\!+\!1\!+\! H^1_\rho) \cup ((c\!+\!1)^2\! -\! H^2_\rho) \right) }_{=:I} \notag
 \end{eqnarray}

We now consider the intersection of the first 3 sets above with $I$ and show that it is contained in  $H^1_\rho \cup ((c+1)^2-H^2_\rho)$. If $H^1_\rho = \emptyset$ or $H^2_\rho=\emptyset$, so by definition $c^2 - H^2_\rho=(c+1)^2-H^2_\rho = \emptyset$ as well, then the  intersections considered below involving such sets are trivially empty. We thus set $\min \emptyset = \infty$ and $\max \emptyset =0$ for the arguments below to apply universally. 
Since $c \geq a \geq 7$ we have that 
$\min (c^2 - H^2_\rho) = c^2 - \max H^2_\rho \geq c^2-2a-1 \geq 2c+1+2a+1 \geq \max(2c+1+\ell, \max(2c+1+H^1_\rho))$, and also $\max (c^2 - H^2_\rho ) \leq c^2 - 1 < (c+1)^2 -(2c+1) \leq \min ((c+1)^2-H^2_\rho)$. Thus $(c^2 -H^2_\rho) \cap I=\emptyset$, so it doesn't contribute to the overall intersection. Since $c\geq a > \sqrt{R_\rho+\ell}+3$, we have $c^2 -R_\rho +1 > c^2 - (a-3)^2 +\ell+1 \geq 2c+2a+2$, where the last inequality is equivalent to the obvious $(c-1)^2-(a-1)^2 +2a -10 +\ell \geq 0$ since $(c-1)^2 \geq (a-1)^2$ and $2a-10 > 0$. Thus $c^2 -R_\rho+1 > \max(2c+1+H^1_\rho)$ and also $c^2-R_\rho+1 > 2c+2a+2 \geq 2c+1+\ell$, and so $[c^2\!-\!R_\rho\!+\!1,(c\!+\!1)^2\!-\!R_\rho] \cap I \subset ((c\!+\!1)^2\! -\! H^2_\rho)$. Putting all of this together equation \eqref{eq:sets} becomes
\begin{eqnarray*}
 k \in \left(  H^1_\rho \cup (c^2 \! - \! H^2_\rho) \cup [c^2\!-\! R_\rho\!+\! 1,(c\!+\!1)^2\!-\!R_\rho] \right)  \cap I 
=\\  \left(  H^1_\rho \cap I \right) \! \cup \! \left( (c^2 \! - \! H^2_\rho)\cap I \right) \! \cup \! \left(  [c^2\!-\! R_\rho\!+\! 1,(c\!+\!1)^2\!-\!R_\rho] \cap I \right) \\
\subset H^1_\rho \cup \emptyset \cup  ((c\!+\!1)^2\! -\! H^2_\rho) 
 \end{eqnarray*}
and we see that $k \notin S_{k+1}$, in contradiction to $k \in S_{k+1}$.
Thus the assumption is wrong and we conclude $S_{c+1} \subset P_{c+1}$, so the induction is complete.
\end{proof}

Let $1^j$ denote the rectangular partition $j \times 1$.
Let $(i,1^j)$ denote the hook partition of $i+j$ with $i$ boxes in the first row and $j+1$ boxes in the first column.
So $(k-j,1^j)$ has $k$ boxes.
The next corollary says that most hooks have positive rectangular Kronecker coefficient.
\begin{corollary}\label{cor:hooks1} 
Let $w \geq h \geq 7$, then $g( (hw-j,1^j ), h \times w, h \times w)>0$ for all $j \in [0,h^2-1]\setminus \{1,2,4,6, h^2-2,h^2-3, h^2-5, h^2-7\}$.
\end{corollary} 
\begin{proof}
We apply Prop.~\ref{prop:extensions} with $\rho = \emptyset$, $a=7$, $R_\rho=1$ and $H^1_{\rho} = \{1,2,4,6\}$ and $H^2_\rho = \{2,3,5,7\}$. The values at $a=7$ are readily verified by direct computation. Then we have for all $b \geq 7$ that $g( (b^2-j,1^j), b \times b, b \times b )>0$ for 
$ j \in [0,b^2-1] \setminus (H^1_\rho \cup (b^2 - H^2_\rho) )$. Let $h=b$ and use the semigroup property once to add the positive triple $((h(w-h)), h \times (w-h), h \times (w-h))$ in order to obtain the statement. 
\end{proof}

Let $(i,1^j+\rho)$ denote the partition that results from $\rho$ by first adding an additional column with $j$ boxes and then adding a top row with $i$ boxes.
The next corollary treats the positivity of hooks that are not covered by Cor.~\ref{cor:hooks1}.

\begin{corollary}\label{cor:columns1246}
Fix $w \geq h  \geq 7 $. We have that $g(\lambda,h \times w, h\times w) >0$ for all $\lambda = ( hw  - j-|\rho|, 1^j + \rho)$ with
$\rho \neq \emptyset$ and $|\rho| \leq 6$ 
for all $j \in [ 1, h^2 - R_\rho]$, where $R_\rho = |\rho|+\rho_1+1$, except in the following cases:
(i) $\rho = (1)$ with $j=2$ or $j=h^2-4$;
(ii) $\rho=(2)$ with $j=2$;
(iii) $\rho =(1^2)$ and $j=1$ or $j=h^2-5$; 
(iv) $\rho=(2,1)$ and $j=1$.
\end{corollary}
\begin{proof}
For all values of $j \leq 6$ we have finitely many partitions $\nu^j := 1^j+ \rho$ of length at most 6 and width at most 7, for which we verify computationally the statement with $h=w=7$ and then by the semigroup property deduce it for $\lambda=\nu^j(hw) , h \times w, h \times w$ with $h,w \geq 7$. 

Now assume that $j>6$, so that $j \geq \ell(\rho)$ and then $j$ falls into the conditions of  Prop.~\ref{prop:extensions}, which we now apply with the following values for $\rho, R_{\rho}, H^1_\rho, H^2_\rho$. 
We have that $\ell(\rho),\rho_1 \leq 6$ and   $R_\rho +\ell(\rho)= |\rho| + \rho_1 +\ell(\rho)+1 \leq |\rho| + |\rho|+2 \leq 14$, so $7>\sqrt{R_\rho+\ell(\rho)} +3$, so we can apply Proposition~\ref{prop:extensions} with initial condition $a=7$, which is verified computationally for the following sets $H^1_\rho$ and $H^2_\rho$. 

(i) When $\rho=(1)$ then $R_\rho = 3$, $H^1_{\rho} = \{2\}$ and $H^2_{\rho} =\{4\}$, we obtain $g( (h^2 -1 - j, 2,1^{j-1}), h \times h, h \times h)>0$ for $j\in [1, h^2-R_\rho] \setminus \{2, h^2 - 4\}$.

(ii) When $\rho=(2)$ then $R_\rho = 5$,
 $H^1_{\rho} =\{2\}$ and $H^2_{\rho} = \emptyset$.

(iii) When $\rho = (1^2)$ then $R_\rho = 4$, $H^1_{\rho} = \emptyset$ and $H^2_\rho = \{5\}$. Note that here we apply Proposition~\ref{prop:extensions} for $j \geq \ell(\rho)=2$, whereas the exceptional $j=1$ was already treated in the computational verification.  

(iv) When $\rho=(2,1)$ and $j\geq 2$, then we apply Proposition~\ref{prop:extensions} with $H^1_\rho=H^2_\rho= \emptyset$, and the case $j=1$ was excluded computationally. 

For all other $\rho$ with $j\geq \ell(\rho)$,  set $H^1_{\rho} = H^2_\rho =\emptyset$ and $R_\rho = |\rho|+\rho_1+1$. The cases $j<\ell(\rho)$ were already treated computationally.

Last, we add the positive triple $(h(w-h), h \times (w-h), h \times (w-h))$ to  $ ( h^2  - j-|\rho|, 1^j+\rho ), h \times h, h \times h$ to obtain the statement. 
\end{proof}

Now we are ready to prove the main positivity theorem.
\begin{theorem}\label{thm:mainpositivityresult}
Let $\nu \notin \exceptions$ and $\ell = max(\ell(\nu)+1,9)$,
$a >  3\ell^{3/2}$, 
$b \geq 3 \ell^2$ and
$|\nu| \leq ab/6$.
Then $g(\nu(ab),a \times b, a \times b)>0$.
\end{theorem}
\begin{proof}
For the ease of notation let $\bar k := k-1$ and $\bar \ell := \ell-1$.
We decompose $\nu$ (Lem.~\ref{lem:partitiondecomposition}) into
$
\nu = \rho + \xi + \sum_{k=2}^\ell x_k (\bar k \times k) + \sum_{k=2}^\ell y_k (\bar k \times 2).
$
If we encounter a block of $s_k := \lfloor a/k\rfloor$ many $\bar k \times k$ rectangles, then we group it
to a large $\bar k\times a_k$ rectangle, where $a_k := k s_k$,
which is $a$ rounded down to the next multiple of $k$.
Formally,
we divide $x_k$ by $s_k$ to obtain $x_k = h_k s_k + t_k$,
where $t_k < s_k$.
So we write
\begin{eqnarray*} 
\nu &=& \textstyle\rho + \xi + \sum_{k=2}^\ell h_k (\bar k \times a_k) \\
&+&\textstyle\sum_{k=2}^\ell t_k (\bar k \times k) + \sum_{k=2}^\ell y_k (\bar k \times 2).
\end{eqnarray*}
We will treat these 5 summands independently.

Using Lem.~\ref{lem:rect_pos} with $s=s_k$ (recall $a_k=k s_k$) we see that
$
g( (\bar k \times a_k)(ak), a \times k, a \times k)>0.
$
Using the semigroup property for the $h_k$ summands we get
\begin{equation}\label{eq:summandI}
g( (\bar k \times (a_k h_k))(h_k a k), a \times (k h_k), a \times (k h_k))>0.
\end{equation}

Using Lem.~\ref{lem:rect_pos} again, this time with $s=t_k$ we see that
\begin{equation}\label{eq:summandII}
g( (\bar k \times (k t_k))(ak), a \times k, a \times k)>0.
\end{equation}

Using Corollary~\ref{cor:hooks1} twice with the semigroup property
(in the case $k \notin\{2,3,5,7\}$)
or using Corollary~\ref{cor:columns1246} we see that
$
g((\bar k \times 2)(2hw),h \times 2w, h \times 2w) > 0
$
for all $h,w \geq 7$, $h \geq \sqrt{ k+7}$, $w \geq \sqrt{ k+7}$.

Choose $h=w:=\max(\lceil\sqrt{\ell(\nu)+8}\rceil,7)$ to obtain
$
g((\bar k \times 2)(2w^2), w \times 2w, w \times 2w) > 0
$
for all $2 \leq k \leq \ell(\nu)$.
Using the semigroup property $y_k$ times we get
$
g((\bar k\times (2 y_k))(2w^2 y_k), w \times (2y_k w), w \times (2y_k w)) > 0
$
and hence by transposition:
$
g((\bar k\times (2 y_k))(2w^2 y_k), (2y_k w) \times w , (2y_k w) \times w ) > 0.
$
Since $y_k <k \leq \ell(\nu)+1$, we have that $a \geq 3\ell^{3/2} \geq 2\ell(\nu) \sqrt{\ell(\nu)+8} \geq2y_k w$ we have
$
g((\bar k\times (2 y_k))(a w), a \times w , a \times w ) > 0.
$
The semigroup property gives
\begin{equation}\label{eq:summandIII}
\begin{split}
&\textstyle g(( \sum_{k=2}^\ell (\bar k\times (2 y_k))  )(a w\bar\ell), \\
&a \times (w\bar\ell), a \times (w\bar\ell) ) > 0.
\end{split}
\end{equation}

The columns of $\xi$ are all distinct and not of length 1, 2, 4, 6.
Thus we can use Corollary~\ref{cor:hooks1} to obtain
$
g( (\bar k\times 1)(w^2), w \times w, w \times w) > 0.
$
Since $\xi$ has at most $\ell-1$ columns the semigroup property gives
$
g( \xi(w^2\bar\ell), w \times (w\bar\ell), w \times (w\bar\ell)) > 0.
$
Transposition gives
$
g( \xi(w^2\bar\ell), (w\bar\ell) \times w, (w\bar\ell) \times w) > 0.
$
Since $a \geq w\bar\ell$ we have
\begin{equation}\label{eq:summandIV}
g(\xi(aw), a \times w, a \times w) > 0.
\end{equation}

For $\rho$ we make the case distinction from Lemma~\ref{lem:partitiondecomposition}.
In cases $(1)$, $(3)$, $(4)$, and $(5)$ a finite calculation shows that $g(\rho(49),7\times 7,7\times 7)>0$.
In case $(2)$ we invoke Cor.~\ref{cor:columns1246} to see that
$g(\rho(wa),a\times w,a\times w)>0$.
In both cases we have
\begin{equation}\label{eq:summandV}
g(\rho(wa),a\times w,a\times w)>0).
\end{equation}

Using the semigroup property on equations \eqref{eq:summandI}, \eqref{eq:summandII}, \eqref{eq:summandIII}, \eqref{eq:summandIV}, and \eqref{eq:summandV} we obtain
\begin{equation*}
g(\nu(aM),a \times M,a \times M) > 0,
\end{equation*}
where $M = \sum_{k=2}^\ell k h_k +  \sum_{k=2}^\ell k + w(\ell-1) + 2w$.
We want to show that $M \leq b$.
Note that
$
h_k \leq \frac{c_k}{a_k} = \frac{c_k}{k\lfloor a/k\rfloor} \leq \frac{c_k}{k(a/k-1)} = \frac{c_k}{a-k} \leq \frac {c_k}{a-\ell}.
$
This can be used to show
$
\sum_{k=2}^\ell k h_k \leq \sum_{k=2}^\ell k {c_k}/({a-\ell}) = (|\nu|+\nu_1)/({a-\ell}).
$
\begin{eqnarray*}
M &=& \textstyle\sum_{k=2}^\ell k h_k +  \sum_{k=2}^\ell k + w(\ell-1) + 2w\\
&\leq&  \frac{|\nu|+\nu_1}{a-\ell} + \ell(\ell+1)/2-1 + w\ell +w \\
&\leq& \textstyle\frac{2|\nu|}{a-\ell} + \ell(\ell+1)/2+ w\ell + w+1\\
&\leq&  \frac{ ab}{3 ( a - \ell) } +(\ell+1)^2 \leq  \frac{3b}{8} \frac{8/9}{1 - \ell/a} + \frac{5\ell^2}{4} \\ 
&\leq & \frac{3b}{8} + \frac{5 b/2}{4} 
\leq b.
\end{eqnarray*}

For the last lines we observe that $1 - \ell/a \geq 1 -\frac{1}{3\sqrt{\ell} } \geq 8/9$, and also that $w < \sqrt{\ell+7}+1\leq \frac{\ell+1}{2}$ and $(\ell+1)^2 \leq 5/4 \ell^2$ for $\ell \geq 9$. 
\end{proof}

\begin{proof}[Proof of Thm.~\ref{thm:secondmain}]
If $\bar\la \in \exceptions$, a finite calculation reveals $a_{\bar\la} = 0$.
Combining this with Prop.~\ref{pro:main} gives $a_\la(d[n])=0$.
If $\bar\la \notin \exceptions$,
we invoke Thm.~\ref{thm:mainpositivityresult} with $\nu=\bar\la$, $\ell = m^2$, 
$a=d$, and $b=n$.
Note that $b \geq 3 m^4 = 3 \ell^2$,
and $a \geq 3 m^3 = 3 \ell^{3/2} $
and $|\bar\la|\leq md = a\sqrt\ell \leq ab/6$.
\end{proof}

\section{Further positivity results: limit coefficients, stretching factor, and semigroup saturation}\label{sec:corollaries}
In the rest of the appendix
we prove the positivity of rectangular Kronecker coefficients for a large class of partitions
where the side lengths of the rectangle are at least quadratic in the length of the partition.
Moreover, we prove that the saturation of the rectangular Kronecker semigroup is trivial,
we show that the rectangular Kronecker positivity stretching factor is 2 for a long first row,
and we completely classify the positivity of
rectangular limit Kronecker coefficients that were introduced by Manivel in 2011.

\subsection{Classification of vanishing limit rectangular Kronecker coefficients}
Recall the definition $\exceptions = \{(1),(2\times1),(4\times1),(6\times1),(2,1),(3,1)\}$ from Thm.~\ref{thm:secondmain}.
Using Thm.~\ref{thm:secondmain} we get a complete classification of all cases in which $a_\rho = 0$.
\begin{corollary}\label{cor:zeroness}
$a_\rho = 0$ iff $\rho \in \exceptions$.
\end{corollary}
\begin{proof}
Given $\rho \notin \exceptions$, $a_\rho>0$ can be seen by choosing large $d$ and $n$ and applying Thm.~\ref{thm:secondmain}.
For $\rho \in \exceptions$, $a_\rho=0$ is a small finite calculation.
\end{proof}

\subsection{Double and triple column hooks}

Here we study Kronecker coefficients for partitions $\lambda = i^k$ and the Kronecker coefficients $g(\lambda(ab), a \times b, a \times b)$ when  $i=2$ and $i=3$. In the case of $i=1$ these were exactly the hooks which were already classified. By that classification and the semigroup property it is readily seen that since $\lambda = i ( 1^k)$, for $k \neq 1,2,4,6, d^2-7, a^2-5, a^2-3, a^2-2$ we have $g(\lambda(ab), a \times b, a \times b)>0$ for $b$ large enough. We now prove that this positivity holds in fact for all $k \in [0,a^2-1]$ when $i>1$.

\begin{proposition}\label{prop:stretchedhooks}
Let $i>1$. For any  $m  \geq 7$ and $k\in[0, m^2 -1 ]$ we have that $g( i(m^2-k, 1^k), m \times (im), m \times (mi) )>0$.
\end{proposition} 
\begin{proof}
 First, note that proposition follows from the hook positivity as long as $k \neq 1,2,4,6, m^2-7, m^2-5, m^2-3, m^2-2$. In the  case when $k=1,2,4,6$  finite calculations for $i=2,3$ and $m=7$ give positive values. For $i\geq 4$ we have that $i = 2i_1+ 3i_2$ for some $i_1,i_2 \geq 0$, and then we apply the semigroup property for $i_1$ many double hooks plus $i_2$ many triple hooks. By that argument, we can always assume that $i \leq 3$. 

So we can assume that $k \in\{ m^2-7, m^2 -5, m^2-3, m^2-2\}$, i.e. $r: = m^2 - k -1 \in [0,6]$ is finite.

 Let $\mu^i[a,b]:= ((k+1)^i, 1^{ir})=((ab-r)^i,1^{ir})$ be its transpose partition.
 Note that $\mu^i[m,m] = ( i(m-k, 1^k) )^t$, so by transposing one of the rectangles the statement to prove is equivalent to showing that 
 $$g( \mu^i[m,m], (im) \times m, m \times (im) )>0.$$
 This will follow from the following claim applied when $a=b=m$:
 
Claim: We have that for all $a,b \geq 6$, $r \leq 7$ and $i \in [2,3]$
 $$g( \mu^i[a,b], (ia) \times b, a \times (bi))>0.$$  
 We prove this claim by induction on $(a,b)$, with initial condition computationally verified for $(a,b)=(7,7)$ for the given finite set of values for $i$ and $r$. 
 Suppose that the claim holds for some values $a=a_0 ,b=b_0 \geq 7$, i.e. we have
 $g( \mu^i[a_0,b_0], (ia_0) \times b_0, a_0 \times (b_0 i))>0$. Consider the triple $(a,b)=(a_0,b_0+1)$. Since $ g( (ia_0), (a_0^i), (a_0^i) )>0$ after transposing the first two partitions and rearranging them we also have that 
 $g( a_0^i, 1^{ia_0}, i^{a_0}) )>0$. Add this triple to $\mu^i[a_0,b_0], (ia_0) \times b_0, a_0 \times (b_0 i)$, applying the semigroup property, we have
 that 
 \begin{align*}
 0 < g( \mu^i[a_0,b_0] + a_0^i , (ia_0) \times b_0 + 1^{ia_0},a_0 \times (b_0 i) + i^{a_0} ) \\
 = g(  ( (a_0 b_0-r)^i +a_0^i,1^{ir} ) , (ia_0) \times (b_0+1), a_0 \times (b_0i + i) )
 \\
 = g( \underbrace{( (a_0 (b_0+1) -r)^i,1^{ir} ) }_{=\mu^i[a_0,b_0+1]}, (ia_0) \times (b_0+1), a_0 \times (b_0+1) i).
 \end{align*}
 This show that the claim holds for $(a_0,b_0+1)$ as well. By the symmetry between $a$ and $b$, we also have the statement for $(a_0+1,b_0)$, so by induction the claim holds for all $(a,b)$, s.t. $a,b \geq 7$. 
 
 Applying the claim with $a=b=m$ completes the proof.

\end{proof}

\subsection{Stretching factor 2}
In \cite{bci:10} it is shown that there exists a stretching factor $i \in \IN$ such that
$g(i\la,d\times (in),d \times (in))>0$.
It is easy to see that the stretching factor $i$ is sometimes larger than~2:
For example take $n=6$, $d=1$ and $\rho=5 \times 1$, then $\la:=\rho(nd)=6 \times 1$ and $g(6 \times 2,6\times 2,6 \times 2)=0$,
but $g(6 \times 3,6\times 3,6 \times 3)=1>0$, so here the stretching factor is 3.
For a long first row the next corollary shows that the stretching factor is always 1 or 2.

\begin{corollary}\label{cor:even}
Fix $m \geq 7$. Let $\rho$ be a partition with $\ell(\rho)<m^2$ in which every row is of even length.
Then $a_\rho(m)>0$.
\end{corollary}
\begin{proof}
Cut $\rho$ columnwise and group pairs of columns of the same length $k$ so that you get partitions $(k \times 2)$.
By Prop.~\ref{prop:stretchedhooks} we have $a_{(k \times 2)}(m)>0$.
Using the semigroup property we get the result.
\end{proof}

\subsection{Trivial saturation of the rectangular Kronecker semigroup}
In \cite{BHI:15} the semigroup of partitions with positive plethysm coefficent is studied.
In analogy we study here the semigroup of partitions with positive rectangular Kronecker coefficient.
For fixed $d$ the partitions $\la$ where $d$ divides $|\la|$ and where $g(\la,d \times (|\la|/d), d \times (|\la|/d))>0$
form a semigroup $S_d$ under addition.
Interpreting these partitions as integer vectors, the real cone $C_d$ spanned by them
coincides with the simplex $\{x \in \IR^{d^2} \mid x_1 \geq x_2 \geq \cdots \geq x_{d^2}\}$,
which is shown in \cite{bci:09}.
The \emph{group $G_d$ generated by} $S_d$ is defined as the set of all differences $\{\la-\mu \mid \la,\mu\in S_d\}$.
The \emph{saturation} of $S_d$ is defined as the intersection $G_d \cap C_d$.
The following corollary shows that the saturation of $S_d$ is as large as possible.
\begin{corollary}\label{cor:kronsaturation}
Let $d \geq 7$. The group $G_d$ contains all partitions $\la$ for which $d$ divides $|\la|$.
\end{corollary}
\begin{proof}
By Prop.~\ref{prop:stretchedhooks} for all $0 \leq k < m^2$ we have $(k \times 3)(3m^2) \in S_d$ and $(k \times 2)(2m^2) \in S_d$.
Subtracting these we obtain $(m^2-k,k \times 1) \in G_d$.
For $1 \leq k < m^2$ we subtract $(m^2-k+1,(k-1) \times 1)$ from $(m^2-k,k \times 1)$ to obtain $v_{k+1}:=(-1,0,\ldots,0,1,0,\ldots,0) \in G_d$,
where the $1$ is at position $k+1$.
Given a partition $\la$ we define $\nu:=\sum_{k=2}^{m^2}\la_k v_k$ and obtain a vector that coincides with $\la$ in every entry but the first:
$\nu_1 = -|\bar\la|$.
We calculate $\la = \mu + j\cdot(d,0,\ldots,0)$ with $j=|\la|/d$.
Since $d$ divides $|\la|$ it follows that $j$ is an integer.
Since $(d,0,\ldots,0)\in S_d$ we have $\la = \mu+j\cdot(d,0,\ldots,0) \in G_d$.
\end{proof}

\section{Exact results for Kronecker coefficients}\label{sec:morepositivity}

Here we provide a complete classification of triples $\lambda, a, b$ with $\lambda$ -- hook or two-column partition, for which the Kronecker coefficient $g(\lambda, a \times b, a \times b)$ is positive and in the course of the proof give certain stronger quantitative relationships between these coefficients. 

\begin{theorem}[Hook positivity]\label{thm:hook_positive}
Assume that $n \geq d$. 
Let $d \geq 7 $. We then have that $g((nd-k,1^k), d \times n, d\times n) >0$  $k \in [0,d^2-1] \setminus \{1,2,4,6, d^2-7,d^2-5, d^2-3, d^1-2\}$ and is 0 for all other values of $k$.

For $d \leq 6$ we have that $g((nd-k,1^k), d \times n, d\times n) =0$ if $k >d^2-1$ or in the following cases:

\begin{tabular}{p{1in}p{4in}}
$d=$ & Values of $k \leq d^2-1$, for which $g_k(d,n)=0$: \\
6 & $\{1,2,4,6,13,22,29, 31, 33, 34\}$\\
5 & $\{1,2,4,6,11,13, 18, 20, 22, 23\}$\\
4 & $\{ 1,2,4,6, 9, 11, 13, 14\}$\\
3 & $\{1,2,4,6,7\}$\\
2& $\{1,2\}$ \\
\end{tabular}

Moreover, we have that $g((nd-k,1^k), d \times n, d\times n) > g((nd-k+2,1^{k-2}), d \times n, d\times n)$ for $k \leq d^2/2$ and the coefficients form a symmetric sequence in $k=0,\ldots,d^2-1$.
\end{theorem}

\begin{remark}
It is immediate to characterize the triples for which the Kronecker coefficient is 1. 
\end{remark}

\begin{corollary}
Fix $d>6$ and let $\rho$ be a partition with $m_i$ columns of length $i$. Then $g(\rho(nd), d \times n, d\times n)>0$ if $m_i \neq 1$ for $i=1,2,4,6$ and $m_i = 0 $ for $i=d^2-7. d^2-5, d^2 - 3, d^2-2$.
\end{corollary}

\begin{proof}
Direct computation shows that $g(((6^2-3i),i^3), 6\times 6, 6 \times 6) >0$ for $i=1,2,4,6$, so by the semigroup property we have $g(nd -3i, i^3), d \times n, d\times n)>0$ for all $d,n \geq 6$. Since every $m_i \geq 2$ is either even, or $3+2a_i$, we have that $(i^{m_i})$ is an even partition or is the sum of $(i^3) + (i^{2a_i})$. By the semigroup property for Kronecker coefficients then we must have $g( (nd - im_i, i^{m_i}), d \times n, d \times n)>0$ for all $m_i \geq 2$ when $i=1,2,4,6$. By Theorem~\ref{thm:hook_positive} for any $m_i$ for the remaining values of $i \geq d^2-1$.

Since $\rho = \sum_i (i^{m_i})$, the statement follows by the Kronecker semigroup property.\end{proof}

Here we give the proof of Theorem~\ref{thm:hook_positive}. 

In order to prove this theorem, we will derive a simple formula for these Kronecker coefficients, following the approaches set in \cite{bla:12,liu:14,pp:14}. For brevity we set $g_k(d,n) = g((nd-k,1^k), d \times n, d\times n) $. 

\begin{proposition}
 We have that the Kronecker coefficients $g((nd-k,1^k), d \times n, d\times n) $ are equal to the number of partitions of $k$ into distinct parts from $\{3,5,\ldots, 2d-1\}$, where without loss of generality by the symmetry of the Kronecker coefficients we assume $d \leq n$. In other words, we have the following generating function identity:
 $$\sum_{k=0}^{nd-1} g_k(d,n) q^k = \prod_{i=2}^{d} (1+q^{2i-1}).$$
\end{proposition}

\begin{proof}
Let $s_{\lambda}$ denote the Schur function indexed by a partition $\lambda$ and let $*$ denote the Kronecker product on the ring of symmetric functions, i.e. given by
$$s_{\lambda} * s_{\mu} = \sum_{\nu} g(\lambda,\mu,\nu) s_{\nu}.$$
For the sake of self-containment we repeat some calculations appearing in \cite{bla:12,liu:14,pp:14}. 

We invoke Littlewood's identity, stating that
$$s_{\lambda}*(s_{\alpha} s_{\beta}) = \sum_{\theta \vdash |\alpha|, \tau \vdash |\beta|} c^{\lambda}_{\theta,\tau} (s_\theta * s_\alpha) (s_\tau *s_\beta).$$

In the case when $\alpha = (1^k)$ and $\beta=(nd-k)$ we have that $s_\theta *s_{1^k} = s_{\theta'}$ and $s_\tau * s_{nd-k} = s_\tau$, where $\theta'$ is the transposed (conjugate) partition of $\theta$. Observe that $s_{1^k} s_{nd-k} = s_{(nd-k,1^k)} + s_{(nd-k+1,1^{k-1})}$. Rewriting the above identity in this case leads to

$$s_{\lambda}*s_{(nd-k,1^k)} +s_{\lambda}*s_{(nd-k+1,1^{k-1} )} = \sum_{\theta \vdash k, \tau \vdash nd-k} c^{\lambda}_{\theta \tau} s_{\theta'} s_{\tau}.$$

Take inner product with $s_{\mu}$ on both sides. Observe that the left-hand side gives two Kronecker coefficients and on the right side we have $\langle s_\mu, s_{\theta'} s_{\tau} \rangle = c^{\mu}_{\theta'\tau}$ by the Littlewood-Richardson rule, so

\begin{equation}\label{eq:Lit_hook}
 g(\lambda,\mu, (nd-k,1^k) ) + g(\lambda,\mu, (nd-k+1,1^{k-1} ) ) =  \sum_{\theta \vdash k, \tau \vdash nd-k} c^{\lambda}_{\theta \tau} c^{\mu}_{\theta'\tau}.
 \end{equation}

Note that when $k=0$ we have that $g(\lambda,\mu, (nd) ) =1$ if $\lambda=\mu$ and $0$ otherwise, and the above identity holds assuming that the term with $(nd - k +1, 1^{k-1})$ is 0 when $k<1$. 

Let $\lambda=\mu=(d\times n)$. As it is not hard to see by the Littlewood-Richardson rule in this case (see e.g. \cite{pp:14}), we have that 
$$c^{ (d\times n)}_{\delta \gamma} = \begin{cases} 1 & \text{ if }\delta_i + \gamma_{d+1-i} = n \text{ for all $i=1,\ldots,d$} \\ 0 & \text{otherwise} \end{cases}.$$
In other words, the rectangular  Littlewood-Richardson coefficient are equal to 1 only when the partitions $\delta$ and $\gamma$ complement each other inside the rectangle. This is also easy to see from the  fact that $c^{\lambda}_{\delta \gamma} = \langle s_{\lambda/\delta}, s_\gamma\rangle$ and in the case of $\lambda=d \times n$, the skew shape, rotated $180^\circ$ is a straight shape, so the corresponding Schur function should be the same as $s_\delta$ to give nonzero inner product. 

Applying these observation to equation~\eqref{eq:Lit_hook} when $\lambda=\mu=d \times n$, we see that the summands on the right-hand side will be nonzero if and only if $\theta$ and $\theta'$ are both the complement of $\tau$ in the $d \times n$ rectangle, so $\theta=\theta'$, and $\theta \subset d \times n$. Since in this case the product of the Littelwood-Richardson coefficients is just 1, the right-hand side is the number of such partitions $\theta$, i.e. 

\begin{equation}\label{eq:hook_recurrence}
g_k(d,n) + g_{k-1}(d,n) = \# \{\theta \, | \; \theta \vdash k, \theta=\theta', \theta \subset d \times n\}
\end{equation}

where $g_k(d,n) = g( (nd -k, 1^k), d \times n, d \times n)$ and we set $g_{-1}(d,n) = 0$, and $g_{k}(d,n)=0$ for all $ k \geq nd$ so the above identity holds for all $k$. 

It is a classical result in combinatorics that self-conjugate partitions are in direct bijection with partitions into distinct odd parts, via $\theta \to ( 2\theta_1-1,2(\theta_2 -1) -1, \ldots)$. The condition $\theta \subset d \times n$ in this case is equivalent to $\theta_1 \leq \min(d,n)=d$. Thus we can rewrite identity~\eqref{eq:hook_recurrence} as the following generating function identity
\begin{equation}\label{eq:gen_fun}
\sum_{k=0}^{\infty} (g_k(d,n) + g_{k-1}(d,n)) q^k = \prod_{i=1}^d (1+q^{2i-1}). 
\end{equation}

Let $G(q) = \sum_{k=0}^{\infty} g_k(d,n)q^k$, then after an index shift the identity implies
$$G(q) +qG(q) = \prod_{i=1}^d (1+q^{2i-1}).$$
Dividing both sides by $(1+q)$ we obtain the generating function for the hook Kronecker coefficients as desired. 

\end{proof}

We invoke the following result from \cite{pp:14} which extends a result in~\cite{alm:85}.

\begin{proposition} Let $b_i := g_i(d,n) + g_{i-1}(d,n)$, i.e. 
$$
 \prod_{i=1}^d  (1+q^{2 i-1})  =  \sum_{j=0}^{d^2} b_j  q^j .
$$
Then, for all  $d\ge 27$, the sequence $(b_{26},\ldots,b_{d^2-26})$ is symmetric and strictly unimodal.
\end{proposition} 
Here strict unimodality means $b_i > b_{i-1}$ for all $26 \leq i \leq \frac{d^2}{2}$ and $b_{i} > b_{i+1}$ for $i >  \frac{d^2}{2}$. 

\begin{proof}[Proof of Theorem~\ref{thm:hook_positive}]
Let $d \geq 27$. 
Since  $b_i = g_i(d,n) + g_{i-1}(d,n)$, we have that $b_i > b_{i-1}$ is equivalent to $g_i(d,n) > g_{i-2}(d,n)$ for $i>26$. 
No term $(1+q^{2i-1})$ for $i>13$ can contribute to the coefficient of $q^k$ for $k\leq 26$, so we have that the  terms in $G(q)$ of order $\leq 26$ are equal to the corresponding terms in 
\begin{align*} \prod_{i=2}^{13} (1+q^{2i-1}) = O(q^{27})+  6*q^{26}+6*q^{25}+6*q^{24}+5*q^{23}+4*q^{22}+4*q^{21}+4*q^{20}+3*q^{19}+3*q^{18} \\
+2*q^{17}+3*q^{16}+2*q^{15}+2*q^{14}+q^{13}+2*q^{12}+q^{11}+q^{10}+q^{9}+q^{8}+q^{7}+q^{5}+q^{3}+1
\end{align*}

So we see that $g_k(d,n) >0$ for $k \in [0,26] \setminus \{1,2,4,6\}$. By the inequality $g_k(d,n) > g_{k-2}(d,n)$ for the  values of $k \geq 26$, and the positivity for $k=24,25$ we obtain the positivity of all other $g_k(d,n)$'s. 

Now let $d \leq 26$. In this case the generating function $G(q)$ can be computed explicitly summarized in the following table:

\begin{tabular}{p{1in}p{4in}}
$d=$ & Values of $k \leq d^2-1$, for which $g_k(d,n)=0$: \\
$7,\ldots,26$ & $\{1,2,4,6, d^2-7, d^2 -5, d^2 - 3, d^2- 2 \}$ \\
6 & $\{1,2,4,6,13,22,29, 31, 33, 34\}$\\
5 & $\{1,2,4,6,11,13, 18, 20, 22, 23\}$\\
4 & $\{ 1,2,4,6, 9, 11, 13, 14\}$\\
3 & $\{1,2,4,6,7\}$\\
2& $\{1,2\}$ \\
\end{tabular}

\end{proof}

While the rectangles so far have been the same, we observe that in most cases when $\lambda=(a^b)$ and $\mu=(c^d)$ are two different rectangles and $\nu=(n-k,k)$ is a two-row, then the Kronecker coefficients is almost always 0.

\begin{proposition}
Let $\lambda=(a^b)$ and $\mu=(c^d)$, where $ab=cd=N$ and $a\neq c$. 
Then
$$g(\lambda,\mu,(N-k,k))=0,$$
for every $k\neq N/2$ and 
$$g(\lambda,\mu,(N/2,N/2)) =1$$
if  $(d-b) | d$ and is $0$ otherwise. 
\end{proposition}
\begin{proof}
Using Littlewood's identity for $\tau = (k)$ and $\theta =(N-k)$ , since $s_k * s_{\alpha}=s_{\alpha}$ and $s_{N-k}*s_{\beta}=s_{\beta}$, we have that
$$s_{\lambda}*(s_k s_{n-k})=\sum_{\gamma \vdash n, \alpha\vdash k, \beta \vdash n-k} c^{\lambda}_{\alpha\beta} c^{\gamma}_{\alpha\beta}s_{\gamma}.$$
As in \cite{pp:14}, the Jacobi-Trudi identity for a two row gives 
$$s_{\nu}=s_ks_{N-k}-s_{k-1}s_{N-k+1}$$
and combining with with the previous identity we have
$$g(\lambda,\mu,\nu)= \sum_{\alpha\vdash k,\beta \vdash n-k} c^{\lambda}_{\alpha\beta}c^{\mu}_{\alpha\beta} - 
\sum_{\alpha\vdash k-1,\beta \vdash N-k+1} c^{\lambda}_{\alpha\beta}c^{\mu}_{\alpha\beta} .$$
We now consider when $c^{\lambda}_{\alpha\beta}c^{\mu}_{\alpha\beta} \neq 0$. It is easy to see, and has been elaborated in \cite{pp:14}, $c^{(a^b)}_{\alpha\beta}=1$ if and only if $\beta$ is the complement of $\alpha$ within $(a^b)$, and is $0$ otherwise. In other words, $\beta_i = a - \alpha_{b+1-i}$ for $i=1,\ldots,b$. At the same time, we need $c^{(c^d)}_{\alpha\beta}\neq 0$ and so $\beta_i = c-\alpha_{d+1-i}$. Assume that $b<d$, so $a>c$.  Since $\alpha,\beta \subset (a^b) \cap (c^d) = (c^b)$, we have $\alpha_j,\beta_j=0$ for $j>b$. So $\alpha_{j}=c$ for $j \leq d-b$.Together, the constraints for $\beta$ give $\alpha_{j} -\alpha_{j+d-b} = a-c$ for all $i=1,\ldots,d$. This now determines $\alpha$ uniquely as $\alpha_{(d-b)i+j} = c- (a-c)i$ for $1 \leq j\leq d-b$ and $i\geq 0$. Since, further, $\alpha_d=0$ and so $\alpha_b=a-c$, we must have $(d-b)|d$ and $(a-c)|c$. Under these conditions it is easy to see that $\alpha=\beta$, so that $k= (ab)/2=(cd)/2$ and then $g(\lambda,\mu,\nu)=1$.
\end{proof}

By transposing the two row partition and one of the rectangles above we reach the following. 
\begin{corollary}\label{cor:kron_even}
Let $n \neq d$ and $\rho = (2^k,1^{nd-2k})$ be a two-column partition of size $nd$. 
If $k=nd/2$ and $(d-n) | d $, then $g( 2^{nd/2}, n \times d, n \times d)=1$, otherwise
$g(\rho, n \times d, n \times d )=0$.
\end{corollary}

\section{Appendix: padding with the first variable}\label{subsec:additionalvar}
Let $E_n$ denote the space of $n^2 \times n^2$ matrices.
In the literature sometimes $(X_{n,n})^{n-m}\per_m$ is called the padded permanent instead of $(X_{1,1})^{n-m}\per_m$.
We present now a simple interpolation argument that shows that it does not matter much which notion we use.
Clearly if $X_{n,n}^{n-m}\per_m \in E_n\det_n$, then also $X_{1,1}^{n-m}\per_m \in E_n\det_n$
by setting $X_{n,n} \leftarrow X_{1,1}$.
The following claim proves the other direction.
\begin{claim}
There exists a function $N=N(n)$ that is polynomially bounded in $n$ such that
if $X_{1,1}^{n-m}\per_m \in E_n\det_n$, then
then $X_{N,N}^{N-m}\per_m \in E_{N}\det_{N}$.
\end{claim}
\begin{proof}
Let $\det_n$ have skew circuits of size $q(n)$ with $q(n)$ polynomially bounded in $n$.
Let $X_{1,1}^{n-m}\per_m \in E_n\det_n$.
Then there exists a size $q(n)$ skew circuit computing $X_{1,1}^{n-m}\per_m$.
The polynomial $\per_m$ is multilinear and we collect terms that involve $X_{1,1}$ using the notation
$\per_m = X_{1,1}P + Q$, where $X_{1,1}$ does not appear in $P$ or $Q$.
Setting $X_{1,1} \leftarrow 1$ in $X_{1,1}^{n-m}\per_m$
we obtain $R_1:=P+Q$ and setting $X_{1,1} \leftarrow 2$ we obtain $R_2:=2^{n-m+1}P+2^{n-m}Q$.
We see that $P = -\frac 1 {2^{n-m}}(2^{n-m}R_1-R_2)$ and $Q=\frac 1 {2^{n-m}}(2^{n-m+1}R_1-R_2)$, which gives size $2q(n)+3$ skew circuits for $P$ and $Q$.
Thus we get a size $N := 2(2q(n)+3)+2=4q(n)+8$ skew circuit for $\per_m = X_{1,1}P + Q$.
Homogenizing with $X_{N,N}$ as the padding variable we see $X_{N,N}^{N-m}\per_m \in E_{N}\det_{N}$.
\end{proof}

\section*{List of notations}
\renewcommand{\arraystretch}{1.2}
\begin{tabular}{ll}
$g(\la,\mu,\nu)$ & Kronecker coefficient: Multiplicity of $[\la]$ in $[\mu]\otimes[\nu]$ \\
$a_\la(d[n])$ & plethysm coefficient: Multiplicity of $\{\la\}$ in $\Sym^d(\Sym^n V)$ \\
$q_\la^m(d[n])$ & multiplicity of $\{\la\}$ in $\IC[\overline{\GL_{n^2}\per_m^n}]$\\
$o_\la^m(d[n])$ & multiplicity of $\{\la\}$ in $\IC[\Gamma_m^n]$, where $\Gamma_m^n$ is the variety of padded polynomials\\
$|\la|$ & number of boxes in the Young diagram of $\la$\\
$\bar\la$ & $\la$ with its first row removed\\
$\la(N)$ & $(N-|\la|,\la)$\\
$a_\rho(d)$ & $g(\rho(nd),n\times d,n\times d)$ for $n \geq |\rho|$\\
$a_\rho$ & $g(\rho(nd),n\times d,n\times d)$ for $n,d \geq |\rho|$\\
$[1,\ell]$ & $\{1,2,\ldots,\ell\}$
\end{tabular}

\end{document}